\documentclass[12pt,reqno]{article}

\usepackage[usenames]{color}
\usepackage{amssymb}
\usepackage{graphicx}
\usepackage{amscd}
\usepackage{bigfoot}
\usepackage{hyperref}

\usepackage{color}
\usepackage{fullpage}
\usepackage{float}

\usepackage{graphics,amsmath,amssymb}
\usepackage{amsthm}
\usepackage{thmtools}

\declaretheorem[refname={theorem,theorems},Refname={Theorem,Theorems}]{theorem}

\declaretheorem[sibling=theorem,name=Lemma]{lemma}
\declaretheorem[sibling=theorem,name=Proposition]{proposition}
\declaretheorem[sibling=theorem,name=Remark]{remark}

\declaretheorem[sibling=theorem,style=definition]{definition}
\declaretheorem[sibling=theorem,style=definition]{example}
\declaretheorem[sibling=theorem,name=Conjecture]{conjecture}

\newcommand{\infw}[1]{\mathbf{#1}}
\newcommand{\TT}[1]{\text{\tt{#1}}}

\newcommand{\Z}{\mathbb{Z}}

\DeclareMathOperator{\PPL}{PPL}

\author{Anna E. Frid, Enzo Laborde, Jarkko Peltom\"aki}
\title{On prefix palindromic length of automatic words}
\begin{document}
\maketitle

\begin{abstract}
The prefix palindromic length $\mathrm{PPL}_{\mathbf{u}}(n)$ of an infinite word $\mathbf{u}$ is the minimal number of concatenated palindromes needed to express the prefix of length $n$ of $\mathbf{u}$. Since 2013, it is still unknown if $\mathrm{PPL}_{\mathbf{u}}(n)$ is unbounded for every aperiodic infinite word $\mathbf{u}$, even though this has been proven for almost all aperiodic words. At the same time, the only well-known nontrivial infinite word for which the function $\mathrm{PPL}_{\mathbf{u}}(n)$ has been precisely computed is the Thue-Morse word $\mathbf{t}$. This word is $2$-automatic and, predictably, its function $\mathrm{PPL}_{\mathbf{t}}(n)$ is $2$-regular, but is this the case for all automatic words?

In this paper, we prove that this function is $k$-regular for every $k$-automatic word containing only a finite number of palindromes. For two such words, namely the paperfolding word and the Rudin-Shapiro word, we derive a formula for this function. Our computational experiments suggest that generally this is not true: for the period-doubling word, the prefix palindromic length does not look $2$-regular, and for the Fibonacci word, it does not look Fibonacci-regular. If proven, these results would give rare (if not first) examples of a natural function of an automatic word which is not regular.
\end{abstract}

\section{Introduction}
A \emph{palindrome} is a finite word $p=p[1]\dotsm p[n]$ such that $p[i]=p[n-i+1]$ for every $i$, like $level$ or $abba$. We consider decompositions, or factorizations, of a finite word as a
concatenation of palindromes. In particular, we are interested in the 
minimal number of palindromes needed for such a decomposition, which we call the \emph{palindromic length} of a word. For example, the palindromic length of $abbaba$ is $3$ since this word is not a concatenation of two palindromes, but $abbaba=(abba)(b)(a)=(a)(bb)(aba)$.


In this paper, we consider the palindromic length of prefixes of infinite words. This function of an infinite word $\infw{u} = \infw{u}[0] \dotsm \infw{u}[n] \dotsm$ is denoted by $\PPL_{\infw{u}}$: here $\PPL_{\infw{u}}(n)$ is defined as the palindromic length of the prefix $\infw{u}[0] \dots \infw{u}[n-1]$ of length $n$ of $\infw{u}$.

The following conjecture was first formulated, in slightly different terms, in a 2013 paper by Puzynina, Zamboni, and the first author \cite{fpz}.

\begin{conjecture}\label{c1}
  For every aperiodic word $\infw{u}$, the function $\PPL_{\infw{u}}(n)$ is unbounded.
\end{conjecture}

In fact, the paper \cite{fpz} contains two versions of the conjecture: one with the prefix palindromic length and the other with the palindromic length of any factor of $\infw{u}$. Saarela later proved the equivalence of these two statements \cite{saarela}.

In the same initial paper \cite{fpz}, the conjecture was proven for the case when $\infw{u}$ is $p$-power-free for some $p$, as well as for a more general case covering almost all aperiodic infinite words. Its proof for all Sturmian words required a special technique \cite{frid}. The full conjecture remains unsolved.

While upper bounds on the prefix palindromic length can be obtained by usual techniques \cite{akmp}, any lower bounds \cite{dlt,li2} or precise formulas for 
$\PPL_{\infw{u}}(n)$ are astonishingly difficult to obtain, except for the following trivial observation.

\begin{remark}
  If an infinite word $\infw{u}$ contains palindromes of length at most $K$, then $\PPL_{\infw{u}}(n)\geq n/K$ for all $n$.
\end{remark}

Up to our knowledge, the only nontrivial previously known infinite word whose prefix palindromic length has been found precisely \cite{tm} is the Thue-Morse word 
with its many beautiful properties \cite{a_sh_ubi}. This sequence is $2$-automatic, and so it was not surprising that its prefix palindromic length is $2$-regular and its first differences are $2$-automatic. Although the prefix palindromic length does not fall into the class of functions of $k$-automatic words which are known to always be $k$-regular \cite{crs}, we are not aware of any natural functions which would not have this property.

In this paper, we explore the limits of the method used for the Thue-Morse word by considering other automatic words. We prove that $\PPL(n)$ is $k$-regular for every $k$-automatic word containing a finite number of distinct palindromes and find this function for the paperfolding word and the Rudin-Shapiro word. At the same time, we also give computational results allowing to conjecture that for the period-doubling word, which contains infinitely many palindromes, the prefix palindromic length is \emph{not} $2$-regular, and for the Fibonacci word, it is not Fibonacci-regular. At the very least, if a regularity exists, it must be very complicated. If in at least one of these examples the function will be proven to be not regular, it would give a first example of a reasonable easily defined function of an automatic word which is not regular.

\section{Automatic words}
Throughout this paper, we use the notation $u[i..j] = u[i]\ldots u[j]$ for a factor of a finite or infinite word $u$ 
starting at position $i$ and ending at $j$. Note that for technical reasons, we start numbering symbols of finite words with $1$ and of infinite words with $0$.

\begin{definition}
  Let $\infw{u}$ be an infinite word. Then we define the \emph{PPL-difference sequence} $d_{\infw{u}}$ of $\infw{u}$ by setting $d_{\infw{u}}(n) = \PPL_{\infw{u}}(n + 1) - \PPL_{\infw{u}}(n)$ for $n \geq 0$. Notice that we always have $\PPL_{\infw{u}}(1) = 1$, and setting $\PPL_{\infw{u}}(0) = 0$ by convention, we get $d_{\infw{u}}(0) = 1$.
\end{definition}

The following lemma was first proven in the 2015 conference version of the paper \cite{shur1} and then generalised by Saarela \cite[Lemma~6]{saarela}.

\begin{lemma}\label{l:difference_bounds}
  For every word $\infw{u}$ and for every $n\geq 0$, we have
  \begin{equation*}
    \PPL_{\infw{u}}(n)-1\leq \PPL_{\infw{u}}(n+1) \leq \PPL_{\infw{u}}(n)+1.
  \end{equation*}
\end{lemma}

Therefore a PPL-difference sequence can only take the values $-1$, $0$, or $1$. 
We prefer to use the alphabet $\{\TT{-}, \TT{0}, \TT{+}\}$ in place of $\{-1,0,1\}$. So, the PPL-difference can be considered as an infinite word over a three-letter alphabet; in particular, this word is useful for many algorithms related to palindromes \cite{shur2}.

As the name suggests, a word $\infw{u} = \infw{u}[0] \dotsm \infw{u}[n] \dotsm$ is called \emph{$k$-automatic} if there exists a deterministic finite automaton $A$ such that every symbol $\infw{u}[n]$ of $\infw{u}$ can be obtained as the output of $A$ with  the base-$k$ representation of $n$ as the input. For the technical details of this definition and for  basic examples, we refer the reader to \cite{a_sh}.  In this paper, we mostly do not use this definition but several equivalent ones. To introduce them, we need more notions.

\begin{definition}
 A \emph{morphism} $\varphi\colon \Sigma^* \to \Delta^*$ is a map satisfying $\varphi(xy)=\varphi(x)\varphi(y)$ for all words $x,y \in \Sigma^*$. Clearly, a morphism is uniquely determined by images of symbols of $\Sigma$ and can be naturally extended to the set of infinite words over $\Sigma$. If there exists a $k$ such that all images of symbols are of length $k$, the morphism is called \emph{$k$-uniform}; a $1$-uniform morphism is called a \emph{coding}.
 
 If for some morphism $\varphi\colon \Sigma^* \to \Sigma^*$ and for a letter $a\in \Sigma$ the image $\varphi(a)$ starts with $a$, then there exists at least one finite or infinite word $u$ starting with $a$ which is a fixed point of $\varphi$, that is, it satisfies the equation $u=\varphi(u)$. If in addition $\varphi$ is $k$-uniform for $k \geq 2$, the fixed point starting with $a$ is unique and is denoted as $\varphi^{\omega}(a)$.
\end{definition}

The following statement is a combination of two results. The case when $\psi$ is a coding is Cobham's theorem \cite{cobham72}, which can also be found in the monograph of Allouche and Shallit \cite{a_sh} as Theorem 6.3.2. The case when $\psi$ is a $m$-uniform morphism for $m>1$ is a combination of Cobham's theorem and Corollary 6.8.3 of the same monograph.

\begin{theorem}\label{t:cobham}
  An infinite word $\infw{u}$ is $k$-automatic if and only if $\infw{u}=\psi(\varphi^\omega (a))$ for some $k$-uniform morphism $\varphi$ and a uniform morphism $\psi$. Moreover, the morphisms can always be chosen so that $\psi$ is a coding.
\end{theorem}

If $\infw{u} = \psi(\varphi^\omega(a))$, where the morphisms $\varphi$ and $\psi$ are not obliged to be uniform, the word $\infw{u}$ is called \emph{morphic}. As the previous theorem shows, the class of morphic words includes all $k$-automatic words for all $k$.

\begin{definition}
  The \emph{$k$-kernel} $\ker_k(\infw{u})$ of an infinite word $\infw{u} = \infw{u}[0] \dotsm \infw{u}[n] \dotsm$ is the set of arithmetic subsequences of $\infw{u}$ with differences of the form $k^e$ and starting positions inferior to the difference:
  \begin{equation*}
    \ker_k(\infw{u})=  \{(\infw{u}[k^e n + b])_{n\geq0} : e \geq 0, 0 \leq b < k^e\}.
  \end{equation*}
 \end{definition}

 An infinite word $\infw{u}$ is $k$-automatic if and only if $\ker_k(\infw{u})$ is finite \cite[Thm.~6.6.2]{a_sh}.

In what follows, we will need and use the equivalent definitions of a $k$-automatic words based on uniform morphisms and on the $k$-kernel.

\begin{example}\label{e:tm}
  For the Thue-Morse word $\infw{t}=0110100110010110\dotsm$, which is $2$-automatic, the three definitions work as follows:
 \begin{itemize}
  \item The definition involving the automaton: the symbol $\infw{t}[n]$, $n=0,1,\ldots$, is $0$ if the number of $1$'s in the binary representation of $n$ is even and $1$ if it is odd.
  \item The definition involving morphisms: $\infw{t}=\sigma^{\omega}(0)$, where $\sigma(0)=01$ and $\sigma(1)=10$; the coding $\psi$ from the formula $\infw{t}=\psi(\sigma^{\omega}(0))$ here is trivial ($\psi(0)=0$, $\psi(1)=1$) and can be omitted.
  \item The definition involving the $2$-kernel: $\infw{t}$ can be described as \emph{the} word starting with $01$ and obtained by alternating the symbols of $\infw{t}$ and of the word $\overline{\infw{t}}=1001\dotsm $ obtained from $\infw{t}$ by exchanging $0$'s and $1$'s. It is not difficult to see that the $2$-kernel of $\infw{t}$ contains only two elements: $\infw{t}$ and  $\overline{\infw{t}}$.
 \end{itemize}

\end{example}

The following definition is closely related to automatic words.

\begin{definition}
  A $\Z$-valued sequence is \emph{$k$-regular} if the $\Z$-module generated by its $k$-kernel is finitely generated.
\end{definition}

This definition implies in particular that $k$-automatic sequences are $k$-regular (we may always assume that a word is over an integer alphabet). A sequence is $k$-automatic if and only if it is a bounded $k$-regular sequence \cite[Thm.~16.1.5]{a_sh}.

Many sequences related to $k$-automatic sequences are $k$-regular, as it follows from an important decidability result by Charlier, Rampersad, and Shallit \cite{crs}. In particular, this is true for the function of factor complexity defined as the number of factors of  length $n$ of the word for each $n$ and for the number of distinct palindromes of length $n$ in the word. In fact, the latter function is even $k$-automatic since it is bounded \cite{abcd}. Thus it is natural to ask if the sequence $\PPL_{\infw{u}}$ is $k$-regular when $\infw{u}$ is $k$-automatic. The next lemma shows that in order to study this question, it suffices to study the $\PPL$-difference sequence.

\begin{lemma}
  Let $\infw{u}$ be an infinite word. Then the sequence $\PPL_{\infw{u}}$ is $k$-regular if and only if the PPL-difference sequence $d_{\infw{u}}$ is $k$-automatic.
\end{lemma}
\begin{proof}
  The set of $k$-regular sequences over $\Z$ is closed under componentwise shift, sum, and difference \cite[Thm.~16.2.1, Thm.~16.2.5]{a_sh}. Therefore $\PPL_{\infw{u}}$ is $k$-regular if and only if $d_{\infw{u}}$ is $k$-regular. By \autoref{l:difference_bounds}, the sequence $d_{\infw{u}}$ is bounded. The conclusion follows from the above-cited fact that a bounded $k$-regular sequence is $k$-automatic \cite[Thm.~16.1.5]{a_sh}.
\end{proof}

The first author studied in \cite{dlt,tm} the PPL-difference sequence $d_{\infw{t}}$ of the Thue-Morse word $\infw{t}$ from \autoref{e:tm} and characterized it as the fixed point of the following $4$-uniform morphism:
\begin{equation*}
  \begin{cases}
    \TT{+} \mapsto \TT{++0-}, \\
    \TT{0} \mapsto \TT{++--}, \\
    \TT{-} \mapsto \TT{+0--}.
  \end{cases}
\end{equation*}

This means in particular that $d_{\infw{t}}$ is $4$-automatic and thus $2$-automatic \cite[Thm.~6.6.4]{a_sh}. Hence $\PPL_{\infw{t}}$ is $2$-regular. This result is so far the only one that completely determines the functions $\PPL_{\infw{u}}$ and $d_{\infw{u}}$ for any nontrivial infinite word $\infw{u}$.

Notice that the result on the Thue-Morse word is not covered by the main result of this paper, since the Thue-Morse word contains infinitely many palindromes. Every prefix of length $4^n$ of the Thue-Morse word is a palindrome, so \autoref{t:main} below does not apply to it.

\section{Automatic first differences}
The following theorem is the main result of this paper.

\begin{theorem}\label{t:main}
  If a $k$-automatic word $\infw{u}$ contains a finite number of distinct palindromes, then the PPL-difference sequence $d_{\infw{u}}$ is $k$-automatic. 
\end{theorem}
\begin{proof}
  Let $p$ be the length of the longest palindrome in $\infw{u}$. Then for every index $n$, the last palindrome in an optimal decomposition of $\infw{u}[0..n-1]$ as a product of palindromes starts at one of the positions $\infw{u}[n-p]$, $\ldots$, $\infw{u}[n-1]$. Thus $\PPL_{\infw{u}}(n)$ is determined by $\PPL(n-p)$, $\ldots$, $\PPL(n-1)$ and the word $\infw{u}[n-p..n-1]$ (we will often omit the subscripts in proofs to improve readability). 
  This simple consideration is a base for the following proposition.
 
 \begin{proposition}\label{p:ddd}
   For every $n$ such that $n\geq m+p$, the number $\PPL_{\infw{u}}(n)$ is uniquely determined by the numbers $\PPL_{\infw{u}}(m)$, $d_{\infw{u}}(m)$, $d_{\infw{u}}(m+1)$, $\ldots$, $d_{\infw{u}}(m+p-1)$, and the word $\infw{u}[m..n-1]$. The number $d_{\infw{u}}(n)$ is uniquely determined by $d_{\infw{u}}(m)$, $d_{\infw{u}}(m+1)$, $\ldots$, $d_{\infw{u}}(m+p-1)$, and the word $\infw{u}[m..n]$.
 \end{proposition}
 \begin{proof}
   Let us prove the first statement. Clearly, for every $i$ such that $i\leq p$, we have
   \begin{equation*}
     \PPL(m+i)=\PPL(m) + d(m) + d(m+1) + \dotsm + d(m+i-1),
   \end{equation*}
   so that $\PPL(m+1)$, $\ldots$, $\PPL(m+p)$ can be reconstructed from $\PPL(m)$, $d(m)$, $d(m+1)$, $\ldots$, $d(m+p-1)$. Now let us proceed by induction on $n\geq m+p$. The preceding computation establishes the base case. Since there are no palindromes in $\infw{u}$ of length greater than $p$, we have
   \begin{equation}\label{e:p12}
     \PPL(n) = \min\{\PPL(n-k)+1 : k = 1,\ldots,p, \text{$\infw{u}[n-k..n-1]$ is a palindrome}\}.
   \end{equation}
   The numbers $\PPL(n-k)+1$ are determined by $\PPL(m)$, $d(m)$, $\ldots$, $d(m+p-1)$, and $\infw{u}[n-p..n-1]$ by hypothesis. The induction step is complete.

   To prove the second statement, we replace $\PPL(m)$ in the previous paragraph by a parameter $P$ and let $\PPL(m+i)-P=D(i)$ for all $i\geq 0$, so that $D(i)=d(m)+d(m+1)+ \dotsm + d(m+i-1)$. Then for $i\leq p$, the number $D(i)$ can be found directly as the sum of the known values of the sequence $d$. Now for $n \geq m+p$, that is, for $i=n-m \geq p$, suppose that the values of $D(j)$ are known for all $j<i$. The claim is true for $i=p$, that is, for $n=m+p$, establishing the base case. For the induction step, it suffices to rewrite \eqref{e:p12} as
   \begin{equation*}
     P+D(i) = \min\{P+D(i-k)+1 : k=1,\ldots,p, \text{$\infw{u}[n-k..n-1]$ is a palindrome}\}
   \end{equation*}
   and to subtract $P$ to obtain $D(i)$ as a function of the previous values of $D$ and the word $\infw{u}[n-p..n-1]$:
   \begin{equation*}
     D(i) = \min\{D(i-k)+1 : k=1,\ldots,p, \text{$\infw{u}[n-k..n-1]$ is a palindrome}\}.
   \end{equation*}
   Now it remains to combine this expression for $i$ and for $i+1$ and to use the formula $d(n)=D(i+1)-D(i)$ to obtain the needed statement.
 \end{proof}

  By \autoref{t:cobham}, we may suppose that $\infw{u}=\psi(\varphi^{\omega}(a))$, where $\psi\colon \Sigma \to \Delta$ is a coding and $\varphi\colon \Sigma \to \Sigma^{k}$ is a $k$-uniform morphism over an alphabet $\Sigma$. Without loss of generality, by passing from $\varphi$ to a power of $\varphi$ if necessary, we may assume that $p<k$. Let
  \begin{equation*}
    \Lambda = \{\psi(\varphi(a)) : a \in \Sigma\}.
  \end{equation*}
  The word $\infw{u}$ is a concatenation of these $\Lambda$-blocks of length $k$, and we consider $\infw{u}$ as $\infw{u} = U[0] \dotsm U[N] \dotsm$ with $U[i] \in \Lambda$.
 
  Consider an occurrence $\infw{u}[m..n]$, where $n \geq m+p$, of a factor $v$ of $\infw{u}$. We define the \emph{type} of this occurrence as the sequence $d_{\infw{u}}[m..m+p-1]$. Clearly, for each word $v$, its occurrences have at most $3^p$ different types; we denote the set of possible types of $v$ by $T(v)$. Notice that the words $U[0]$, $U[1]$, $\ldots$ have types because their lengths are greater than $p$.
 
  The following proposition is a direct corollary of \autoref{p:ddd}

  \begin{proposition}\label{p:dddd}
    For every $N>0$, the type of the occurrence $U[N]$ is determined by the word $U[N]$, the word $U[N-1]$, and the type of $U[N-1]$.
  \end{proposition}
 
  This proposition can be interpreted as follows: given a word $U[0] \dotsm U[N] \dotsm$ and the type of $U[0]$, we can uniquely determine the types of $U[1]$, $U[2]$ and so on, and thus, due to \autoref{p:ddd}, find the PPL-difference sequence $d$. The process can be described by a transducer with 
  \begin{itemize}
    \item set of states $\{(A,t) : A\in \Lambda, t \in T(A))\} \cup \{S\}$, where $S$ is a special starting state;
    \item input alphabet $\Lambda$;
    \item output set $\{\TT{-}, \TT{0}, \TT{+}\}^k$; and
    \item set of transitions defined as follows:
          \begin{itemize}
            \item The starting transition marked as $U[0] | d_{\infw{u}}[0..k-1]$ goes from $S$ to the state $(U[0], d[0..p-1])$;
            \item A state $(A,t)$ is linked to a state $(B,t')$ by a transition marked as $B | w$ if a $\Lambda$-block $A$ of type $t$ is followed by a $\Lambda$-block $B$ of type $t'$ in $\infw{u}$ and the respective block of length $k$ in $d$ is $w$ (meaning in particular that $t'$ is a prefix of $w$).
          \end{itemize}
  \end{itemize}
  The transitions are well defined due to Propositions \ref{p:ddd} and \ref{p:dddd}, and the number of states is finite as $\# \Lambda \leq \# \Sigma$ and each word in $\Lambda$ has at most $3^p$ types. It is evident that the transducer describes the construction of $d$ from the $\Lambda$-blocks of $\infw{u}$.

  Since the sequence of $\Lambda$-blocks of $\infw{u}$ is $k$-automatic by the construction, we see that the sequence $d$ is obtained by feeding it to a uniform transducer (a uniform transducer outputs only words of common length). By a theorem of Cobham \cite{cobham72} (see also \cite[Thm.~6.9.2]{a_sh} and the discussion preceding it), a uniform transduction of a $k$-automatic sequence is again $k$-automatic, so we conclude that $d$ is $k$-automatic. Notice that if we replaced $k$ by its power, we still obtain the same conclusion as a sequence is $k^\ell$-automatic if and only it is $k$-automatic \cite[Thm.~6.6.4]{a_sh}.
\end{proof}

\begin{example}
  Consider the $2$-automatic fixed point $\infw{u} = \mu^\omega(a) = abbcbccabccacaab \dotsm$ of the morphism
  \begin{equation*}
    \mu\colon \begin{cases} a \mapsto ab, \\ b \mapsto bc, \\c \mapsto ca.\end{cases}
  \end{equation*}
  It is not difficult to see that the longest palindromes in $\infw{u}$ are of length $3$, so, in order to construct the transducer of the proof of \autoref{t:main}, we consider $\infw{u}$ as a fixed point of the $4$-uniform morphism 
  \begin{equation*}
    \mu^2\colon \begin{cases} a \to abbc, \\ b \to bcca, \\c \to caab.\end{cases}
  \end{equation*}
  For the alphabet $\Lambda$, we now have $\Lambda = \{A, B, C\}$ where $A=abbc$, $B=bcca$, $C=caab$. The first values of $\PPL_{\infw{u}}(n)$ starting from $n=0$ are $0,1,2,2,3,3,3,4,5$, and thus the sequence $d_{\infw{u}}$ starts with $\TT{++0+00++}$. Hence the first transition of the transducer is
  \begin{equation*}
    S \xrightarrow{A | \TT{++0+}} (A, \TT{++0}).
  \end{equation*}
  The next transition should describe the first differences in $B$ which follows an occurrence of $A$ with type $\TT{++0}$. It can be checked that it is
  \begin{equation*}
    (A,\TT{++0}) \xrightarrow{B | \TT{00++}} (B, \TT{00+}).
  \end{equation*}
  Continuing to consider blocks and their types in their order of appearance in $\infw{u}$, we can analogously find that every symbol of $\Lambda$ can have four types $\TT{++0}$, $\TT{00+}$, $\TT{0+0}$, $\TT{-+0}$. Thus the transducer has $13$ states. The possible transitions from $A$ are the following:
  \begin{align*}
    (A,\TT{00+}) &\xrightarrow{A | \TT{++0+}} (A,\TT{++0}), & (A,\TT{0+0}) &\xrightarrow{B | \TT{00++}} (B,\TT{00+}), \\
    (A,\TT{00+}) &\xrightarrow{B | \TT{-+0+}} (B,\TT{-+0}), & (A,\TT{++0}) &\xrightarrow{B | \TT{00++}} (B,\TT{00+}), \\
    (A,\TT{00+}) &\xrightarrow{C | \TT{0+0+}} (C,\TT{0+0}), & (A,\TT{-+0}) &\xrightarrow{B | \TT{00++}} (B,\TT{00+}).
  \end{align*}
  In particular, a block $A$ of any type except for $\TT{00+}$ can be followed only by the block $B$ of type $\TT{00+}$.
  
  The remaining transitions are obtained by changing the letters in the above transitions according to the cycle $A \to B \to C \to A$ since the initial morphism $\mu$ is symmetric with respect to this cycle. For example, from the transition
  \begin{equation*}
    (A,\TT{00+}) \xrightarrow{A | \TT{++0+}} (A,\TT{++0})
  \end{equation*}
  we obtain in this fashion the transition
  \begin{equation*}
    (B,\TT{00+}) \xrightarrow{B | \TT{++0+}} (B,\TT{++0}).
  \end{equation*}
  This gives a total of $19$ transitions. To be completely rigorous, we should prove that no additional states and transitions exist. Suppose the opposite and consider the first transition which is not as above. Let it be a transition from $(A, \text{\tt0+0})$ to $(A, t)$ for some $t$ (any other combination can be considered analogously). The first time this transition is taken must be preceded by a transition from the list, that is, by the transition $(B, \text{\tt00+}) \xrightarrow{} (A, \text{\tt0+0})$. Hence $BAA$ should be a factor of $\mathbf{u}$, but it is easy to check that this is not the case. Similarly if there is a transition $(A, \text{\tt0+0}) \xrightarrow{} (C, t)$, then we find that $\mathbf{u}$ should contain the forbidden factor $BAC$, and so on.

%
  It can be shown that the output of the transducer equals the infinite word $\psi_{\infw{u}}(\varphi^\omega_{\infw{u}}(s))$, where
  \begin{equation*}
    \varphi_{u}\colon \begin{cases} s \mapsto su, \\u \mapsto eu, \\ e \mapsto du, \\ d \mapsto hu, \\ h \mapsto eu \end{cases}
  \end{equation*}
  and
  \begin{equation*}
    \psi_{u}\colon
    \begin{cases}
      s,e \mapsto \TT{++0+}, \\ d \mapsto \TT{0+0+}, \\
      u   \mapsto \TT{00++}, \\ h \mapsto \TT{-+0+}.
       \end{cases} 
  \end{equation*}
  Here the symbols $s$, $d$, $e$, $u$, $h$ mean respectively the starting block $s$ of $\infw{u}$, the situation when the next block of $\infw{u}$ is down ($d$), equal ($e$) or up ($u$) to the previous block according to the cyclic order $A<B<C<A$. And $h$ (for ``high'') stands for the situation when the block is exactly the third in an ascending sequence of blocks.
  \end{example}

  In this example, we managed to construct the morphisms for $d_{\infw{u}}$ because we understand the underlying structure. Unfortunately, Cobham's theorem used in the proof of \autoref{t:main} only gives a hyperexponential bound on the number of the states of an automaton generating $d_{\infw{u}}$. Hence the theorem itself does not give a practical way to find $d_{\infw{u}}$ and the associated morphisms. In what follows, we consider two well-known examples and find their prefix palindromic length ``by hand''.

\section{Classic examples}

The examples considered in this section, namely, the paperfolding word and the Rudin-Shapiro word, are closely related to each other, and are known to contain a finite number of distinct palindromes \cite{allouche97}.

\subsection{Paperfolding word}
Recall that the paperfolding word $\infw{u}_{pf}$ is the $2$-automatic word 
\begin{equation*}
  \infw{u}_{pf} = \psi (\varphi_{pf}^{\omega}(a)) = 0010011000110110\dotsm,
\end{equation*}
where 
\begin{equation*}
  \varphi_{pf}\colon \begin{cases} a \mapsto ab, \\ b \mapsto cb, \\c \mapsto ad, \\ d \mapsto cd, \end{cases}
\end{equation*}
and the coding $\psi$ is defined as $\psi(a)=\psi(b)=0$, $\psi(c)=\psi(d)=1$. 

The longest palindromes in the paperfolding word are of length $13$, so \autoref{t:main} can be applied to it: its first difference sequence $d_{pf}$ is $2$-automatic. The blocks considered in the proof of \autoref{t:main} could be of length $16$, since it is the smallest integer power of $2$ which exceeds  the length of the longest palindrome. However, to simplify the transcducer, it is more convenient to consider  blocks of length $64$. 
\begin{theorem}\label{t:pf}
  The sequence $d_{pf}$ over the alphabet $\{\TT{-}, \TT{0}, \TT{+}\}$ is equal to $d_{pf}=\gamma_{pf}(\mu^{\omega}_{pf}(a_0))$, where 
 \begin{equation*}
   \mu_{pf}\colon
   \begin{cases}
    a_0 \mapsto a_0 b_a, \\
    a_b \mapsto a_b b_a, \quad b_a \mapsto c_b b_c, \quad c_b \mapsto a_b d_a, \quad d_a \mapsto c_b d_c,\\  
    a_d \mapsto a_d b_a, \quad b_c \mapsto c_d b_c, \quad c_d \mapsto a_d d_a, \quad d_c \mapsto c_d d_c
   \end{cases}
 \end{equation*}
 and
 \begin{equation*}
   \gamma_{pf}\colon
   \begin{cases}
     a_0 \mapsto& \hspace{-0.7em} \TT{\footnotesize{+0+0-+0+000-++0-+-}}P\TT{\footnotesize{00+00+0-+000+000+00000-0++0-+0-}}, \\
     a_b \mapsto& \hspace{-0.7em} \TT{\footnotesize{0++-0+00+00-0+0000}}P\TT{\footnotesize{00+00+0-+000+000+00000-0++0-+0-}},\\  
     a_d \mapsto& \hspace{-0.7em} \TT{\footnotesize{0+00000+00000+000-}}P\TT{\footnotesize{00+00+0-+000+000+00000-0++0-+0-}},\\
     b_a \mapsto& \hspace{-0.7em} \TT{\footnotesize{0++-0+00+00-0+0000}}P\TT{\footnotesize{+-+0-0+000+000+0+-+0-000+000+0-}},\\
     b_c \mapsto& \hspace{-0.7em} \TT{\footnotesize{+00+-00+0000+0000-}}P\TT{\footnotesize{+-+0-0+000+000+0+-+0-000+000+0-}},\\
     c_b \mapsto& \hspace{-0.7em} \TT{\footnotesize{0++-0+00+00-0+0000}}P\TT{\footnotesize{00+00+0-+000+000+00000-0++0-+00}},\\   
     c_d \mapsto& \hspace{-0.7em} \TT{\footnotesize{0+00000+00000+000-}}P\TT{\footnotesize{00+00+0-+000+000+00000-0++0-+00}},\\
     d_a \mapsto& \hspace{-0.7em} \TT{\footnotesize{0++-0+00+00-0+0000}}P\TT{\footnotesize{+-+0-0+000+000+0+-+0-000+000+00}},\\     
     d_c \mapsto& \hspace{-0.7em} \TT{\footnotesize{+00+-00+0000+0000-}}P\TT{\footnotesize{+-+0-0+000+000+0+-+0-000+000+00}}
    \end{cases} 
 \end{equation*}
 with $P = \TT{0+00+00-0++-0+0}$.
\end{theorem}
\begin{proof}
  Let $\infw{v}$ be the fixed point $\varphi_{pf}^\omega(a)$ of $\varphi_{pf}$ and $\infw{w}$ be the fixed point $\mu_{pf}^\omega$ of $\mu_{pf}$ starting with $a_0$. The word $\infw{v}$ is obtained from $\infw{w}$ by the identification $a_0, a_b, a_d \mapsto a$, $b_a, b_c \mapsto b$, $c_b, c_d \mapsto c$, $d_a, d_c \mapsto d$. The subscript of a letter occurring in $\infw{w}$ indicates that the letter (after identification) in $\infw{v}$ is preceded by the letter indicated by the subscript, that is, $a_b$ is corresponds to $a$ preceded by $b$ in $\infw{v}$ etc. The letter $a_0$ simply corresponds to the first occurrence of $a$ in $\infw{v}$.

  We know by \autoref{t:main} that a transducer $T$ mapping $\infw{u}_{pf}$ to $d_{pf}$ exists. Here we set the parameter $k$ of the proof of \autoref{t:main} to equal $2^6$. This means that $T$ outputs blocks of length $64$. Write $\infw{u}_{pf} = U[0] U[1] \dotsm$ as a concatenation of $\Lambda$-blocks $U[i]$. For the claim, it suffices to prove that the output of $T$ on $U[0] U[1] \dotsm U[n]$ equals $\gamma_{pf}(\infw{w}[0..n])$ for all $n$.
  
  The factors of $\infw{v}$ of length $2$ appear in its prefix of length $13$. This means that the prefix of $\infw{u}_{pf}$ of length $13 \times 2^6$, which is a concatenation of $\Lambda$-blocks, contains all possible adjacent $\Lambda$-blocks at least once. We can directly check that $\gamma_{pf}(\infw{w}[0..12])$ coincides with the prefix of $d_{pf}$ of length $13 \times 2^6$ meaning that $\gamma_{pf}(\infw{w}[0..12])$ equals the output of $T$ on $U[0] \dotsm U[12]$. Let us now make the following observation. The prefix of length $18$ of each $\gamma_{pf}$-image is followed by the word $P = \TT{0+00+00-0++-0+0}$ of length $15$. Since the longest palindrome in $\infw{u}_{pf}$ has length $13$, \autoref{p:ddd} implies that for $n = 1, \ldots, 12$, the type of $U[n]$ depends on $P$ and $U[n - 1]$, not on the type of $U[n - 1]$. Since $P$ occurs in the same position in every $\gamma_{pf}$-image, we see that the type of $U[n]$ depends only on $U[n - 1]$.

  Let $k \geq 12$ be such that the type of $U[n]$ depends only on $U[n - 1]$ and that the output of $T$ on $U[0] \dotsm U[n]$ matches $\gamma_{pf}(\infw{w}[0..n])$ for all $n = 1, \ldots, k$. By \autoref{p:ddd}, the type of $U[n+1]$ is determined by $U[n]$ and its type. Since $T$ outputs $\gamma_{pf}(\infw{w}[n])$ when reading $U[n]$ and $\gamma_{pf}(\infw{w}[n])$ contains $P$ at position $18$ independently of the letter $\infw{w}[n]$, it follows from \autoref{p:ddd} that the type of $U[n+1]$ depends only on $U[n]$. Since $k \geq 12$ and all factors of $\infw{v}$ of length $2$ appear in its prefix of length $13$, there exists $t \leq 11$ such that $U[t] = U[n]$ and $U[t+1] = U[n+1]$. The output of $T$ on the transition $U[n] \xrightarrow{} U[n+1]$ must match that of $U[t] \xrightarrow{} U[t+1]$ because $T$ is deterministic and the type of the $\Lambda$-block is irrelevant in both cases. Therefore $T$ outputs $\gamma_{df}(\infw{w}[t+1])$ when reading $U[n+1]$. It now suffices to show that $\infw{w}[n+1] = \infw{w}[t+1]$ in order to conclude by induction that $d_{pf} = \gamma_{pf}(\infw{w})$.

  We have $U[i] = \psi( \varphi_{pf}^6( \infw{v}[i] ))$ for all $i$. It is straightforward to verify that $\psi$ is injective on the set of $\Lambda$-blocks and that $\varphi_{pf}^6$ is injective, so we deduce from the equalities $U[n] = U[t]$ and $U[n+1] = U[t+1]$ that $\infw{v}[n] = \infw{v}[t]$ and $\infw{v}[n+1] = \infw{v}[t+1]$. From the first paragraph of the proof, we infer that $\infw{w}[n+1] = \infw{w}[t+1]$. The claim follows.
\end{proof}

\subsection{Rudin-Shapiro word}
The Rudin-Shapiro word $\infw{u}_{rs}$ is the $2$-automatic word 
\begin{equation*}
  \infw{u}_{rs}=\psi (\varphi_{rs}^{\omega}(a)) = 00010010000111010\dotsm,
\end{equation*}
where 
\begin{equation*}
  \varphi_{rs}\colon \begin{cases} a \to ab, \\ b\to ac, \\ c \to db, \\d \to dc, \end{cases}
\end{equation*}
and the coding $\psi$ is defined by $\psi(a)=\psi(b)=0$, $\psi(c)=\psi(d)=1$. 

The longest palindromes in the Rudin-Shapiro word are of length $14$, so \autoref{t:main} can be applied to it: its first difference sequence $d_{rs}$ is $2$-automatic. The following theorem describes it.

\begin{theorem}
  The sequence $d_{rs}$ over the alphabet $\{\TT{-}, \TT{0}, \TT{+}\}$ is equal to
  \begin{equation*}
    d_{rs}=\gamma_{rs}(\mu^{\omega}_{rs}(A)),
  \end{equation*}
  where 
  \begin{equation*}
    \mu_{rs}\colon \begin{cases} A \to AB, \\ B \to CD, \\ C \to EB, \\ D \to ED, \\ E \to CB \end{cases}
  \end{equation*}
  and
  \begin{equation*}
    \gamma_{rs}\colon
    \begin{cases}
      A \mapsto& \hspace{-0.7em} \TT{\footnotesize{+00+00000-++00-++00-+0+00+00+00+-0+00-+00+-0+0+0-0+0}}P, \\
      B \mapsto& \hspace{-0.7em} \TT{\footnotesize{0+0-0++-00+0+0-++00-+0+00+00+00+0+0-0++-00+0+0-+000+}}P, \\
      C \mapsto& \hspace{-0.7em} \TT{\footnotesize{-0+00-+00+-0+0+00+00-++-+00+000+0-+000+-+0-0+0+0-0+0}}P, \\
      D \mapsto& \hspace{-0.7em} \TT{\footnotesize{-0+00-+00+-0+0+00+00-++-+00+000+0+0-0++-00+0+0-+000+}}P, \\
      E \mapsto& \hspace{-0.7em} \TT{\footnotesize{0+0-0++-00+0+0-++00-+0+00+00+00+-0+00-+00+-0+0+0-0+0}}P.
    \end{cases} 
  \end{equation*}
  with $P = \TT{0-+00+00+00+}$.
\end{theorem}
\begin{proof}
  As previously for the paperfolding word, we define a new morphism $\nu_{rs}$ obtained from $\varphi_{rs}$ by adding to each letter information on the preceding one:
  \begin{equation*}
     \nu_{rs}\colon
      \begin{cases}
        \begin{tabular}{l l l l}
          $a_0 \to a_0b_a$;  &                    &                     & \\ 
          $a_b \to a_c b_a$; & $b_a \to a_b c_a$; & $c_a \to d_b  b_d$; & $d_b \to d_c c_d$;\\  
          $a_c \to a_b b_a$; & $b_d \to a_c c_a$; & $c_d \to d_c b_d$;  & $d_c \to d_b c_d$.
        \end{tabular}
      \end{cases}
  \end{equation*}
  The morphism $\varphi_{rs}$ and its fixed point $\infw{v}$ are obtained from $\nu_{rs}$ and its fixed point $\infw{w}$ by the identification $a_0, a_b, a_c \mapsto a$, $b_a, b_d \mapsto b$, $c_a, c_d \mapsto c$, $d_b, d_c \mapsto d$.

  We proceed as in the proof of \autoref{t:pf}. We set the parameter $k$ of the proof of \autoref{t:main} to equal $2^6$. Write $\infw{u}_{rs} = U[0] U[1] \dotsm$ as a concatenation of $\Lambda$-blocks $U[i]$. All factors of $\infw{v}$ of length $2$ appear in its prefix of length $14$, so all adjacent $\Lambda$-block appear in the prefix of $\infw{u}_{rs}$ of length $14 \times 2^6$. Taking the prefix of length $14 \times 2^6$ of $d_{rs}$, we observe that it coincides with the word $\delta_{rs}(\infw{w}[0..13])$ where
  \begin{equation*}
    \delta_{rs}\colon
    \begin{cases}
      a_0      \mapsto& \hspace{-0.7em} \TT{\footnotesize{+00+00000-++00-++00-+0+00+00+00+-0+00-+00+-0+0+0-0+00-+00+00+00+}}, \\
      b_a,c_d  \mapsto& \hspace{-0.7em} \TT{\footnotesize{0+0-0++-00+0+0-++00-+0+00+00+00+0+0-0++-00+0+0-+000+0-+00+00+00+}},\\
      a_b,d_c  \mapsto& \hspace{-0.7em} \TT{\footnotesize{-0+00-+00+-0+0+00+00-++-+00+000+0-+000+-+0-0+0+0-0+00-+00+00+00+}},\\
      c_a,b_d  \mapsto& \hspace{-0.7em} \TT{\footnotesize{-0+00-+00+-0+0+00+00-++-+00+000+0+0-0++-00+0+0-+000+0-+00+00+00+}},\\
      a_c,d_b  \mapsto& \hspace{-0.7em} \TT{\footnotesize{0+0-0++-00+0+0-++00-+0+00+00+00+-0+00-+00+-0+0+0-0+00-+00+00+00+}}.           
    \end{cases}
  \end{equation*}
  Each $\delta_{rs}$-image of a letter ends with the word $\TT{0-+00+00+00+}$ of length $12$. This word $P$ is shorter than the longest palindrome in $\infw{u}_{rs}$, so we cannot directly deduce that the type of the block $U[n]$ depends only on $U[n-1]$. By \autoref{p:ddd}, the number $d_{rs}( (n-1) 2^6)$ depends on the previous $14$ values of $d_{rs}$ that correspond to a palindrome ending at position $(n-1) 2^6$ of $\infw{u}_{rs}$. We claim that such a palindrome has length at most $12$. This implies that $d_{rs}( (n-1) 2^6)$ is determined by the previous $12$ values of $d_{rs}$. If such a palindrome has length greater than $12$, it must be of length $14$ as $\infw{u}_{rs}$ contains no palindromes of length $13$. Two of the $\Lambda$-blocks end with $110100011101$ and the remaining two end with $001011100010$. It is straightforward to see that neither suffix can be covered by a palindrome of length $14$ in the required way. Thus the palindrome has length at most $12$. A similar argument can be repeated for the number $d_{rs}( (n-1) 2^6 + 1)$. Since each $\delta_{rs}$-image ends with $P$ of length $12$, we deduce by \autoref{p:ddd} that the type of $U[n]$ depends only on $U[n-1]$ not on its type. We may now repeat the arguments of the proof of \autoref{t:pf} and conclude that $d_{rs} = \delta_{rs}(\infw{w})$ (indeed $\varphi_{rs}$ is injective and $\psi$ is injective on the set of $\Lambda$-blocks).

  To prove the theorem, it remains to notice the symmetry in $\delta_{rs}$ and identify $b_a, c_d $ as $B$, $a_b,d_c$ as $C$, $c_a,b_d$ as $D$, $a_c,d_b$ as $E$. After renaming $a_0$ as $A$, we see that $\nu_{rs}^\omega(a_0)$ equals $\mu_{rs}^\omega(A)$ after this identification. Thus $\delta_{rs}(\nu_{rs}^\omega(a_0)) = \gamma_{rs}(\mu_{rs}^\omega(A)$ and the claim follows.
\end{proof}

\section{Computational results and conjectures}
This section contains results of computational experiments which thus do not give any theorems but only conjectures. For a fast computation of the prefix palindromic length, we used an implementation \cite{rosetta} of the Eertree data structure \cite{shur1}; see also \cite{shur3} for related algorithms.

\subsection{Period-doubling word}\label{s:pd}
\autoref{t:main} and the result for the Thue-Morse word allow to conjecture that the $\PPL$-difference sequence $d_{\infw{u}}$ of a $k$-automatic word is always $k$-automatic. The following example, however, suggests that this is not the case.

The period-doubling word $\infw{u}_{pd}$ is the $2$-automatic word 
\begin{equation*}
  \infw{u}_{pd}=\varphi_{pd}^{\omega}(a)=	abaaabababaaabaa\dotsm,
\end{equation*}
where 
\begin{equation*}
  \varphi_{pd}\colon \begin{cases} a \to ab, \\ b\to aa. \end{cases}
\end{equation*}
Clearly, it contains infinitely many palindromes, including all prefixes of length $2^n-1$. 
Thus \autoref{t:main} does not apply to it.

One way to show that a sequence is not automatic is to prove that it has superlinear factor complexity function \cite[Thm.~10.3.1]{a_sh}. We computed the factor complexity of $d_{pd}$ for lengths up to $1000$ and found that, disregarding some initial values, the values closely follow a line with slope $200$. To find evidence for the nonautomaticity of $d_{pd}$, we thus need to propose another experiment.

This approach does not work here because 

In our computational experiment, we estimate the cardinality of the $2$-kernel of the PPL-difference sequence $d_{pd}$ of $\infw{u}_{pd}$. If $d_{pd}$ is $2$-automatic, its $2$-kernel must be finite. We estimate the number of its elements as follows. 

Let $m \geq 1$. Consider a sequence $(d_{pd}[2^e n + b])_n$ from the $2$-kernel of $d_{pd}$ and compute its prefix $d_{e,b}$ such that $2^e n + b \leq 4^m$. Only finitely many different words $d_{e,b}$ are nonempty: in particular, all such words of length at least $2$ correspond to $e<2m$, so there are finitely many parameters to consider. Then we exclude from the set of words $d_{e,b}$ those which are proper prefixes of another word of this set. Let $k_m$ be the number of nonempty words $d_{e,b}$ that remain. Then, clearly, the $2$-kernel of $d_{pd}$ contains at least $k_m$ elements. 

The following table collects the values of $k_m$ for $m = 1, \ldots, 11$.

{\small
\begin{center}
\begin{tabular}{ |c||c|c|c|c|c|c|c|c|c|c|c| } 
  \hline
 $4^m$ & $4$ & $4^2$&$4^3$&$4^4$&$4^5$&$4^6$&$4^7$& $4^8$&$4^9$&$4^{10}$&$4^{11}=4194304$\\ 
 \hline
 $k_m$& 2& 9& 22& 66& 145& 297& 584& 1046& 1816& 3047& 5051\\
 \hline
 $k_{m}/k_{m-1}$&&4.5&2.444&3.0&2.197&2.048&1.966&1.791&1.736&1.678&1.658\\
  \hline
\end{tabular}
\end{center}
}

Our data thus indicates that the $2$-kernel of $\infw{u}_{pd}$ contains at least $5051$ distinct sequences. Moreover, a four times longer prefix gives at least $1.65$ times larger $2$-kernel, and the ratio decreases too slowly to conjecture that it would tend to $1$. This makes an impressive contrast with all the previous examples where the size of the kernel rapidly stabilizes. Based on this, we formulate the following conjecture.

\begin{conjecture}
  The sequence $d_{pd}$ of the period-doubling word $\infw{u}_{pd}$ is {\rm not} $2$-automatic, and so the prefix palindromic length $\PPL_{pd}(n)$ of $\infw{u}_{pd}$ is not $2$-regular.
\end{conjecture}

Another way to show that a sequence is not automatic would be to prove that it has superlinear factor complexity function \cite[Thm.~10.3.1]{a_sh}. However, it does not look to be the case: we computed the factor complexity of $d_{pd}$ for lengths up to $500$ and found that, except for initial values, the values closely follow a straight line with slope about $200$. So, we could suggest that even if the PPL-difference function of the period-doubling word is not $2$-automatic, it can still be morphic.

\subsection{Fibonacci word}
The Fibonacci word $\infw{u}_{f}=abaababaabaab\dotsm$ is the fixed point $\varphi_{f}^\omega(a)$ of the morphism
\begin{equation*}
  \varphi_{f}\colon \begin{cases} a \to ab, \\ b \to a. \end{cases}
\end{equation*}
The Fibonacci word is a classic example of an infinite word; it is not $k$-automatic for any $k$ but is \emph{Fibonacci-automatic} in the sense which we explain below.

As usual, we define the Fibonacci numbers by the recurrence relation $F_0 = 0$, $F_1 = 1$, and $F_n = F_{n-1} + F_{n-2}$ for $n \geq 2$. Every positive integer $n$ can be uniquely expressed as $n = \sum_{0 \leq i \leq r} a_{i} F_{i+2}$ with $a_i \in \{0, 1\}$, $a_r =1$, and $a_i a_{i+1} = 0$ for $0 \leq i < r$. In this case, we call the word $a_r a_{r-1} \dotsm a_0$ the \emph{Fibonacci representation} of $n$ and use the notation $(n)_F = a_r a_{r-1} \dotsm a_0$. For example, we have $(3)_F = 100$ and $(12)_F = 10101$. We also fix $(0)_F=0$. 

As is well-known, $\infw{u}_f[n]=b$ if and only if $(n)_F$ ends with $1$; in the opposite case, we have $\infw{u}_f[n]=a$. Thus every symbol of the Fibonacci word can be computed from the Fibonacci representation of its index by a simple automaton. This means that the Fibonacci word is \emph{Fibonacci-automatic}. In general, an infinite word $\infw{x}$ is Fibonacci-automatic if there exists a deterministic finite automaton $A$ such that every symbol $\infw{x}[n]$ is the output of $A$ with input $(n)_F$. 
Many functions of the Fibonacci word are known to be Fibonacci-automatic or Fibonacci-regular; for the definition and discussions of Fibonacci-regular sequences, see  \cite{mss-1,dmss}.

Analogously to a $k$-kernel for $k$-automatic sequences, we define the \emph{Fibonacci-kernel} of a sequence $\infw{w}$ as follows. For every finite word $s \in \{0,1\}^*$, define $(i_s)$ as the increasing sequence of all numbers $n$ such that $(n)_F$ ends with the suffix $s$. For example, $(i_\varepsilon)=0,1,2,\ldots$, $(i_0)=0,2,3,5,7,\ldots$, $(i_1)=1,4,6,9,12\ldots$, and $(i_{11})$ is empty since the Fibonacci representation cannot contain two consecutive $1$'s.

Now we define a sequence $\infw{w}(s)$ as the subsequence of $\infw{w}$ with indices from $(i_s)$, namely, $\infw{w}(s)=\infw{w}[i_s[0]]\infw{w}[i_s[1]]\infw{w}[i_s[2]]\dotsm$. At last we define the \emph{Fibonacci-kernel} of $\infw{w}$ as the set of nonempty sequences $\infw{w}(s)$ for all $s \in \{0,1\}^*$.

For example, the Fibonacci-kernel of the Fibonacci word $\infw{u}_f$ consists of three elements: the Fibonacci word $\infw{u}_f=\infw{u}_f(\varepsilon)$ itself and the sequences $aa\dotsm a \dotsm = \infw{u}_f(0)$ and $bb \dotsm b \dotsm = \infw{u}_f(1)$. Indeed, we have $\infw{u}_f(p0)=aa\dotsm  a \dotsm = \infw{u}_f(0)$ and $\infw{u}_f(p1)=bb\dotsm b \dotsm = \infw{u}_f(1)$ for every finite word $p$ (or the sequences $\infw{u}_f(p0)$ and $\infw{u}_f(p1)$ are empty).


Notice that the Fibonacci-kernel of an infinite word always contains the empty sequence because $11$ does not occur in Fibonacci representations. We largely ignore this fact.

Analogously to the proof for $k$-automatic words, it can be shown that a sequence is Fibonacci-automatic if and only if its Fibonacci-kernel is finite. 


The existing family of decidability results on Fibonacci-automatic words \cite{mss-1,dmss} is mostly analogous to the $k$-automatic case. It would be interesting to find an example of a reasonable function of the Fibonacci word which takes a finite number of values and is not Fibonacci-automatic. It seems that the PPL-difference sequence $d_f$ of the Fibonacci word is a good candidate for that.

Similar to \autoref{s:pd}, we consider words determined by the (nonempty) sequences of the Fibonacci-kernel of $d_{f}$ and the prefix of $d_{f}$ of length $|\varphi_{f}^{3m}(a)|$ for $m = 1, 2, \ldots$. Let again $k_m$ be the number of the corresponding nonempty words that are not prefixes of each other. Our computations give the following values for $k_m$ for $m = 1, \ldots, 8$.

\begin{center}
\begin{tabular}{ |c||c|c|c|c|c|c|c|c| } 
 \hline
 $|\varphi_{f}^{3m}(a)|$ & $5$ & $21$   & $89$   & $377$  & $1597$  & $6765$ & $28657$ & $121393$ \\
 \hline
 $k_m$                   & $3$ & $11$   & $31$   & $88$   & $207$  & $504$  & $1139$  & $2377$ \\
 \hline
 $k_{m}/k_{m-1}$         &     & $3.67$ & $2.82$ & $2.85$ & $2.35$ & $2.43$ & $2.26$  & $2.09$ \\
  \hline
\end{tabular}
\end{center}

While this evidence is not as strong as in the case of the period-doubling word, we conclude that the Fibonacci-kernel of $d_{f}$ has at least $2377$ elements and the kernel does not seem to stabilize. We make the following conjecture.

\begin{conjecture}
  The sequence $d_{f}$ of the Fibonacci word $\infw{u}_{f}$ is \emph{not} Fibonacci-automatic, and so the prefix palindromic length $\PPL_{f}(n)$ of $\infw{u}_{f}$ is not Fibonacci-regular.
\end{conjecture}

We remark that, as in the case of the period-doubling word, the sequence $d_{f}$ seems to have linear factor complexity based on the first $1000$ values. All Fibonacci-automatic words have linear factor complexity \cite[Thm.~3.4]{parry}.

All Fibonacci-automatic words have linear factor complexity \cite[Thm.~3.4]{parry}, so, as for the period-doubling word, a superlinear factor complexity function would give another evidence supporting the conjecture. However, the sequence $d_{f}$ seems to have linear factor complexity based on the first $500$ values and can thus still be morphic even if not Fibonacci-automatic.

\section{Conclusion}
In this paper, we have proven a general theorem on the prefix palindromic length of automatic words containing finitely many distinct palindromes and considered in detail two particular cases when this theorem is applicable. These results were somehow predictable since they state that a reasonable function of a $k$-automatic word is $k$-regular. What is more surprising is the computational evidence that in some other situations this is not the case: it seems that there exist simple $k$-automatic words, such as the period-doubling word, such that their prefix palindromic length is not $k$-regular. If proven, this result would enrich the whole theory of $k$-regularity.

\end{document}